\documentclass[conference]{IEEEtran}

\usepackage[noadjust]{cite}
\usepackage{graphicx}
\usepackage{float}
\usepackage{amsmath}
\usepackage{amssymb}
\usepackage{amsthm,amssymb}
\usepackage{color}
\usepackage{listings}
\usepackage[tight,footnotesize]{subfigure}

\newtheorem{theorem}{Theorem}[section]

\newtheorem{proposition}[theorem]{Proposition}

\begin{document}

\title{A Statistical Modelling and Analysis of PHEVs' Power Demand in Smart Grids}

\author{\IEEEauthorblockN{Farshad Rassaei, Wee-Seng Soh and Kee-Chaing Chua \\}
\IEEEauthorblockA{Department of Electrical and Computer Engineering\\
National University of Singapore, Singapore\\
Email: \{f.rassaei, weeseng, eleckc\}@nus.edu.sg}
}
\maketitle

\begin{abstract}

Electric vehicles (EVs) and particularly plug-in hybrid electric vehicles (PHEVs) are foreseen to become popular in the near future. Not only are they much more environmentally friendly than conventional internal combustion engine (ICE) vehicles, their fuel can also be catered from diverse energy sources and resources. However, they add significant load on the power grid as they become widespread. The characteristics of this extra load follow the patterns of people's driving behaviours. In particular, random parameters such as arrival time and driven distance of the vehicles determine their expected demand profile from the power grid. In this paper, we first present a model for uncoordinated charging power demand of PHEVs based on a stochastic process and accordingly we characterize the EV's expected daily power demand profile. Next, we adopt different distributions for the EV's charging time following some available empirical research data in the literature. Simulation results show that the EV's expected daily power demand profiles obtained under the uniform, Gaussian with positive support and Rician distributions for charging time are identical when the first and second order statistics of these distributions are the same. This gives us useful insights into the long-term planning for upgrading power systems' infrastructure to accommodate PHEVs. In addition, the results from this modelling can be incorporated into designing demand response (DR) algorithms and evaluating the available DR techniques more accurately.      

\end{abstract}

\IEEEpeerreviewmaketitle

\section{Introduction}
The current state-of-the-art in information technology (IT) and data processing are going to be employed extensively in smart grids \cite{ipakchi_grid_2009}. Widespread deployment of advanced metering infrastructure (AMI) enables real-time and two-way information exchange between demand side users and the electric utility. This evolution affects all different segments of the grid including generation side, transmission, distribution, as well as the demand side. \par 

Traditionally, the utility designs and installs the power grid's infrastructure such that it can provide power to users' adverse daily power demand profiles similar to that shown in Fig. \ref{f2}. This power demand profile has a significant peak-to-average ratio (PAR) that can potentially reduce the power grids' efficiency and incur exorbitant costs for developing the power grid's infrastructure, i.e., increasing the power generation, transmission, and distribution capacity of the power grid. This extra capacity is just to serve the power demand of users during peak-time periods. Therefore, this drawback has motivated intensive research on strategies that can utilize the existing power grid more efficiently so that more consumers can be accommodated and served without developing new costly infrastructure. The main objective of these strategies is to make the demand responsive \cite{mohsenian-rad_autonomous_2010}. Similar power efficiency concerns have become crucially important for supporting larger number of tenants in green cloud data centres \cite{Dalvandi}.

\par

Demand response (DR) is predicted to become even more important as the use of new electricity-hungry appliances such as plug-in electric vehicles (PEVs) or plug-in hybrid electric vehicles (PHEVs) is becoming more widespread. Typically, on charging mode, they can double the average dwelling's energy consumption, with current PHEVs consuming 0.25-0.35 kWh of energy for one mile of driving \cite{van_roy_apartment_2014}. \par

\begin{figure}
    \centering
    \includegraphics[width=\columnwidth]{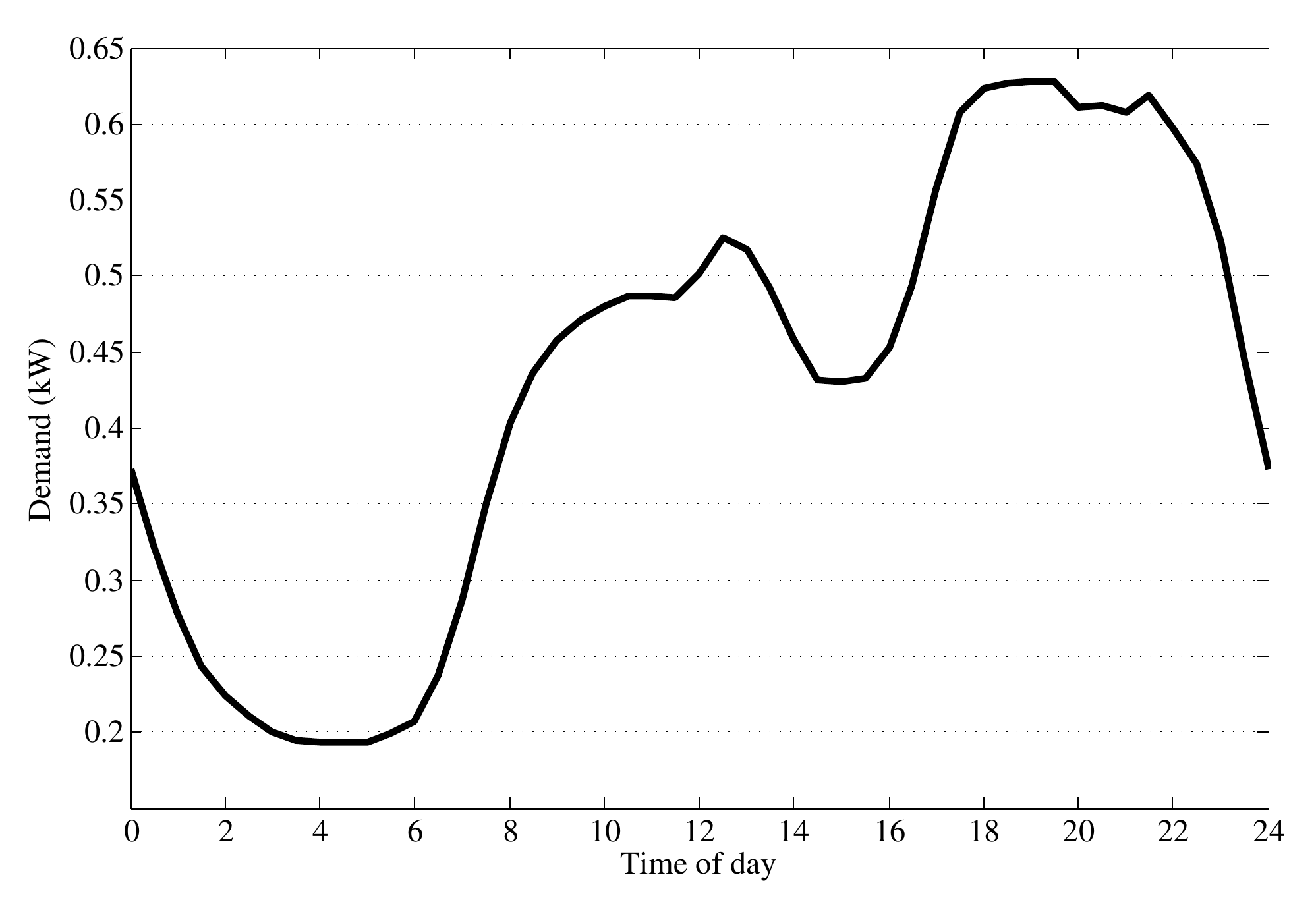}\vspace{-1em}
    \caption{Annual mean daily power demand profile for domestic electricity use in UK ~\cite{richardson2010domestic}.}\vspace{-0.85em}
    \label{f2}
\end{figure}

On the other hand, PEVs have several important advantages compared to internal combustion engine (ICE) vehicles. Not only do they have lower maintenance and operation costs, they also produce little or even no air pollution and greenhouse gases in locales where they are being used \cite{li_modeling_2012}. Above that, they offer valuable flexibility as their fuel can be catered from diverse sources and resources, e.g., nuclear energy and wind power \cite{shireen_plug-hybrid_2010}.

 However, in spite of their vast advantages, the market size of PEVs has been slower than expected as their adoption faces several barriers. One key reason is the extra cost of their batteries. In addition, the shortage of recharging infrastructure causes \textit{range anxiety} for pure electric vehicles' drivers. But, plug-in hybrids resolve the latter problem for pure electric vehicles, by having a combustion engine which works as a backup when the batteries are depleted, yielding to comparable driving ranges for PHEVs to conventional ICE cars \cite{callaway_achieving_2011}. 

\par 


Although it makes sense to envisage the number of electric cars increasing, it is hard to see that the electricity infrastructure capacity growing with the same rate concurrently. Thus, the ramification of introducing a large number of PHEVs into the grid has become an important avenue for research in recent years \cite{shao_challenges_2009}. First, we need to ask how uncoordinated charging, i.e., the battery of the vehicle either starts charging as soon as plugged in or after a user-defined delay, can affect the existing power grid. Next, we need to ask, considering this demand as a worst-case scenario, how we can satisfy it efficiently when we have information exchange capability and intelligence in a smart grid. \par  

There are several prior literature on modelling the impact of uncoordinated charging of PHEVs. However, most of them require much detailed information about passenger car travel behaviour, e.g., \cite{grahn_phev_2014} and \cite{lee_synthesis_2011}. Not only are the models mostly complicated and very test-oriented, but the sensitivity of the PHEVs' charging load to different parameters is not also clear. Moreover, most of them do not provide expected daily power demand due to EVs, particularly when EVs are charged in households rather than in charging stations. For instance, \cite{Spatial} provides a spatial and temporal model of electric vehicles charging demand for fast charging stations situated around highway exits based on known traffic data. In \cite{grahn_phev_2014}, a utilization model is proposed based on type-of-trip. The authors in \cite{li_modeling_2012} have used random simulation and statistical analysis to fit a distribution for the overall charging demand of PHEVs mainly for probabilistic power flow (PPF) calculations. In \cite{alizadeh_scalable_2014}, the daily load profile is modelled by using queuing theory and the approach is suitable mainly for accurate short-time load forecasting. \par

Furthermore, since PHEVs are considered as the main component of the residential flexible electricity demand, numerous researches have been carried out for PHEVs' DR, e.g., \cite{clement_coordinated_2009} and \cite{Clement}. Additionally, their storage capacity can be used for improving the power grid's reliability, e.g., in terms of frequency control \cite{moghadam_randomized_2013}. But, the main drawback in most of these demand response works is that they do not consider the inherent randomness of this demand in the first place. \par 

Therefore, in this paper, we present a stochastic model for uncoordinated charging power demand of a typical PHEV by formulating it as a stochastic process based on the arrival time and driven distance of the vehicles. Moreover, we derive PHEV's expected daily power demand according to this model for arbitrary random distributions of arrival time and charging time. This gives us useful insights into the long-term planning for upgrading the power systems' infrastructure to accommodate PHEVs. In addition, the results from this modelling can be incorporated into designing DR algorithms and evaluating the available DR techniques more accurately. \par

The rest of this paper is organized as follows. Section \ref{SM} provides the system model. Statistical analysis is addressed in section \ref{Sta}. Numerical results and simulations are represented in section \ref{SR}. Finally, section \ref{Con} concludes this paper.

\section{System Model} \label{SM}

In this section, we describe the energy system model and introduce the layout of this study. Fig. \ref{f1} represents a basic power system model where multiple energy customers share one energy source retailer or an aggregator \cite{mohsenian-rad_autonomous_2010} and \cite{mohsenian-rad_optimal_2010}. Consumers' total load consists of two different types of load; flexible load and inflexible load (see Fig. \ref{f3}). Loads which need \textit{on-demand} power supply (e.g., refrigerators) are considered as inflexible, whereas loads that can \textit{tolerate} some delays in power supply (e.g., PHEVs) are assumed as flexible loads \cite{mohsenian-rad_optimal_2010}.

\begin{figure}
	\centering
	\includegraphics[width=\columnwidth]{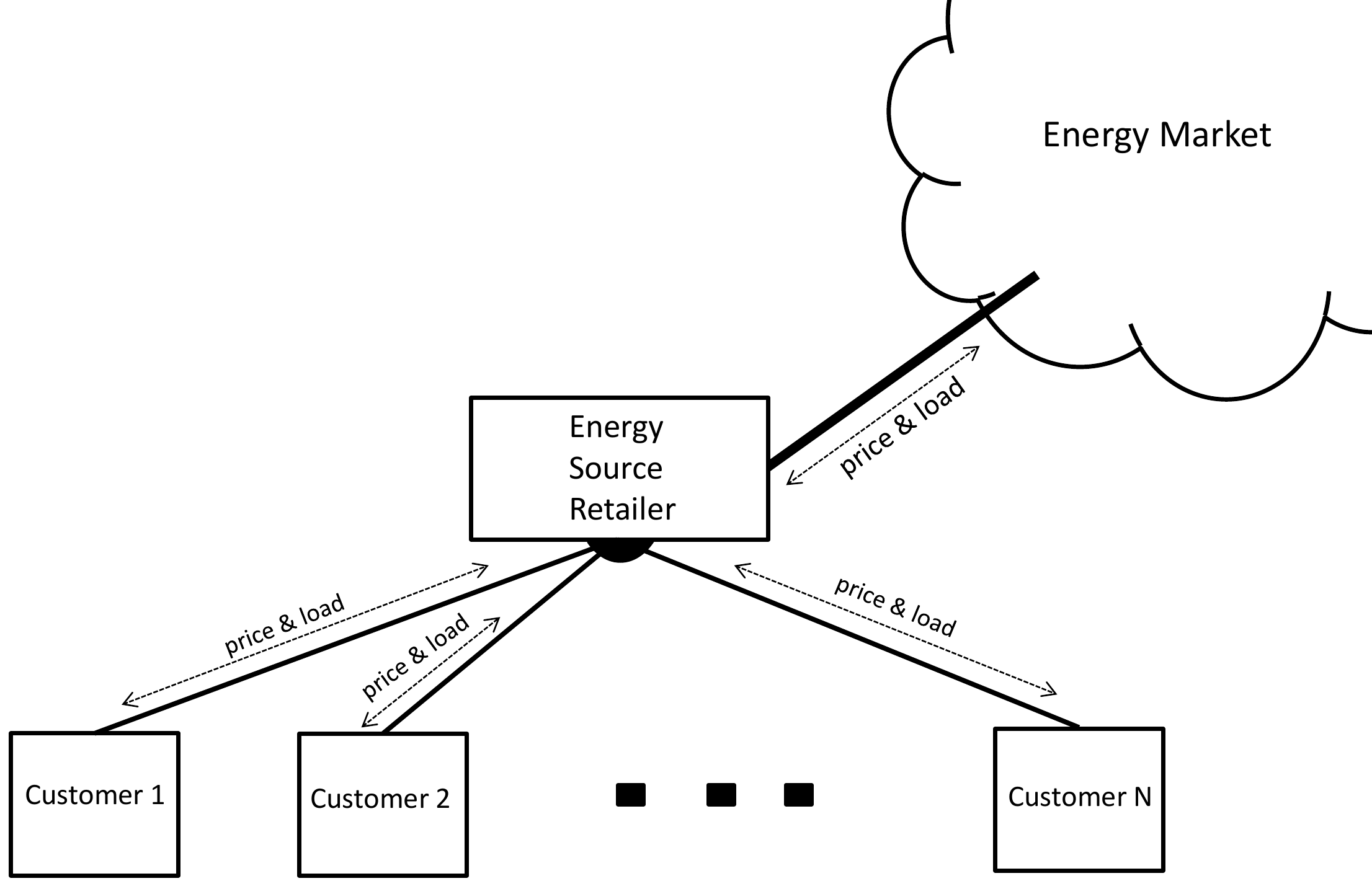} 
	\caption{Basic model of a smart energy system comprised of multiple load customers which share one energy source retailer or an aggregator.}
	\label{f1}
\end{figure}

 \begin{figure}
     \centering
     \includegraphics[width=\columnwidth]{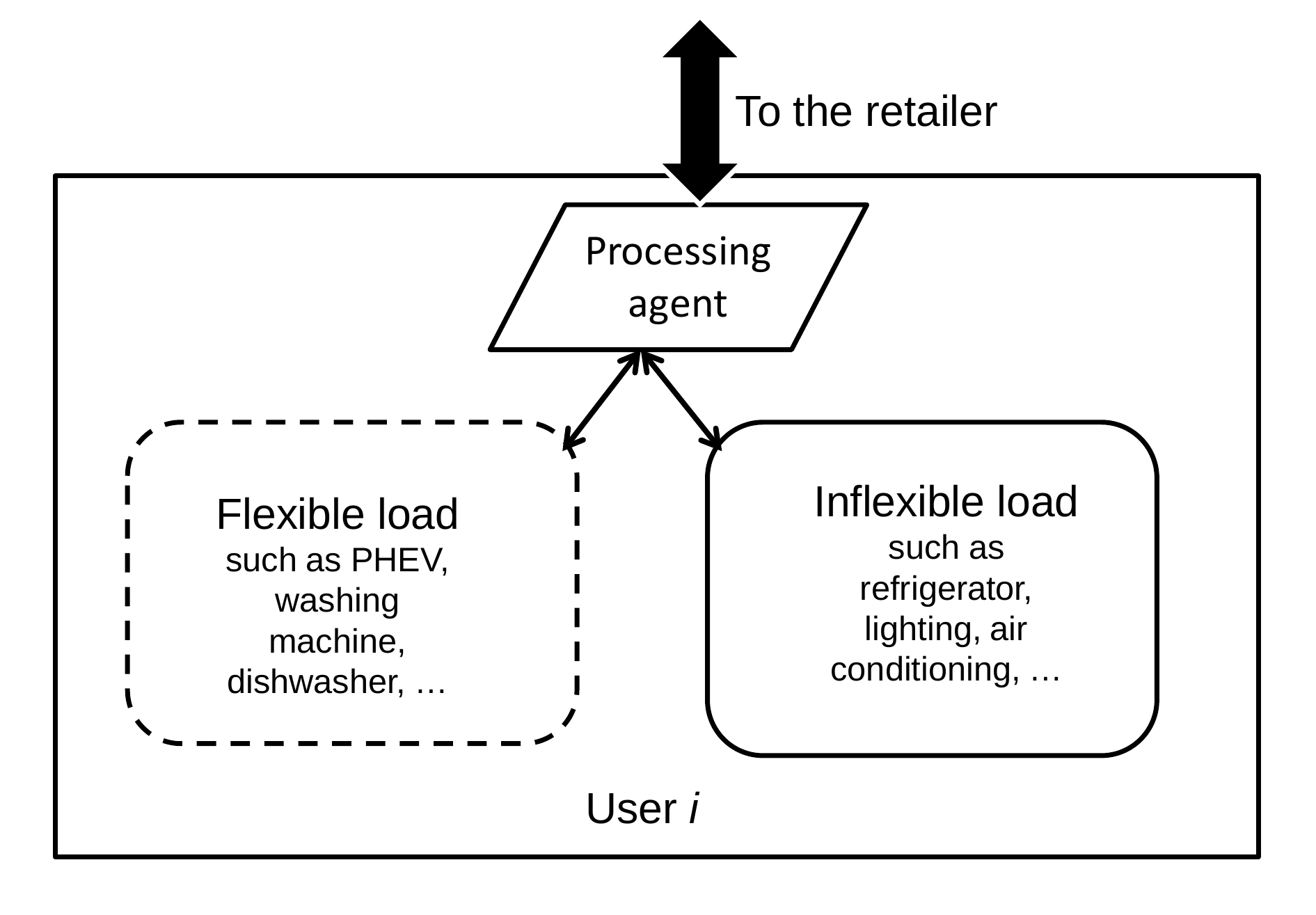} \vspace{-2em}
     \caption{Load segregation of a user according to power demand flexibility.}
     \label{f3}
 \end{figure} 
   
Fig. \ref{f4} displays the demand flexibility of a flexible appliance for different users. A certain job, ordered by user $j$, may take time $T_j$ to be completed. Moreover, the users set not only the desired job but also the deadline by which the job should be accomplished. In this case, we may recognize the following three random variables for a generic flexible appliance:

\begin{itemize}

\item \textbf{Start Time} shows the time when the user lets the power grid connect to the appliance, and can potentially start delivering energy.  

\item \textbf{Operating Time} indicates the time interval required for accomplishing a certain job, e.g., the ordered charging levels and modes (fast charging or slow charging) for PHEVs, which differs from one user to another. 

\item \textbf{End Time} represents the deadline specified by the user for accomplishing the task of the appliance.   

\end{itemize}

\par 

Hence, in general, we need to take into account this randomness when we investigate the overall behaviour of the system. Moreover, to design and analyse DR techniques more accurately, we should consider this stochasticity which comes from the patterns of people's living behaviours and appliance specifications.

\par 
  \begin{figure}
  	\centering
  	\includegraphics[width=\columnwidth]{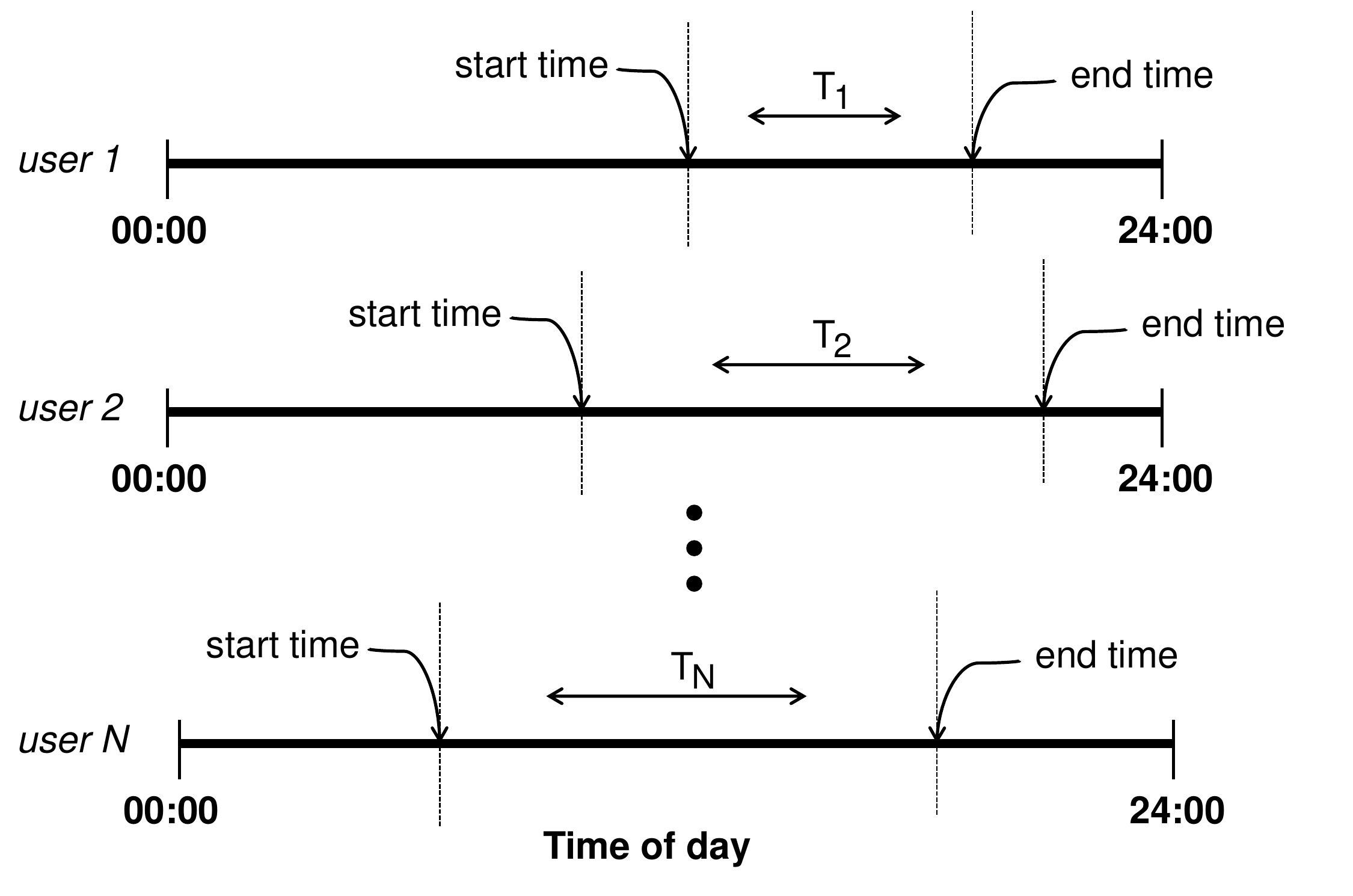}
  	\caption{Time setting for accomplishing a certain job on an appliance for different users during a day.}
  	\label{f4}
  \end{figure}

Therefore, in general, we can formulate the uncoordinated power consumption for an appliance operating a particular job as follows:  


	\begin{equation}
	x(t) \triangleq
	\begin{cases}
	a  &   t_0\leq t < t_0+T  \\
	0       & otherwise \\
	\end{cases}
	\label{xt}
	\end{equation}

\noindent{where we consider instantaneous power consumption as a random variable $a$ and assume that power consumption in standby mode is negligible. Additionally, $T$ and $t_0$ are the operation time and the job's start time, respectively. These parameters are random in general (see Fig. \ref{f6}).

In addition, here, we are mainly interested in knowing the daily power consumption profiles, i.e., the power consumption behaviour throughout a typical 24-hour day. Therefore, we calculate (\ref{xt}) in modulo 24-hours and then project the results onto a 24-hour day. In this case, some realizations of the stochastic process defined in (\ref{xt}) can be displayed as shown in Fig. \ref{f7}. This figure shows (\ref{xt}) for ten different users in a bar graph with one hour time granularity.

  \begin{figure}
      \centering
      \includegraphics[width=\columnwidth]{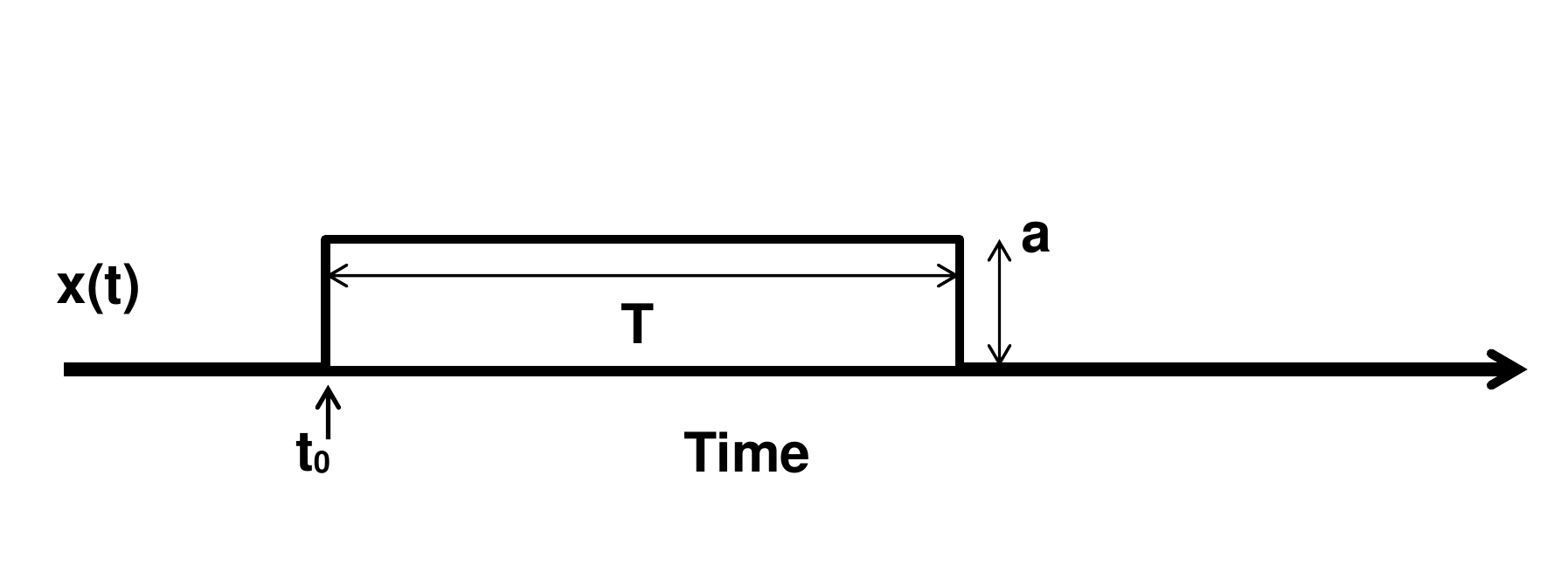}\vspace{-2em}
      \caption{Demonstration of a typical form of $x(t)$.}\vspace{+1em}
      \label{f6}
  \end{figure}

  \begin{figure}
  	\centering
  	\includegraphics[width=\columnwidth]{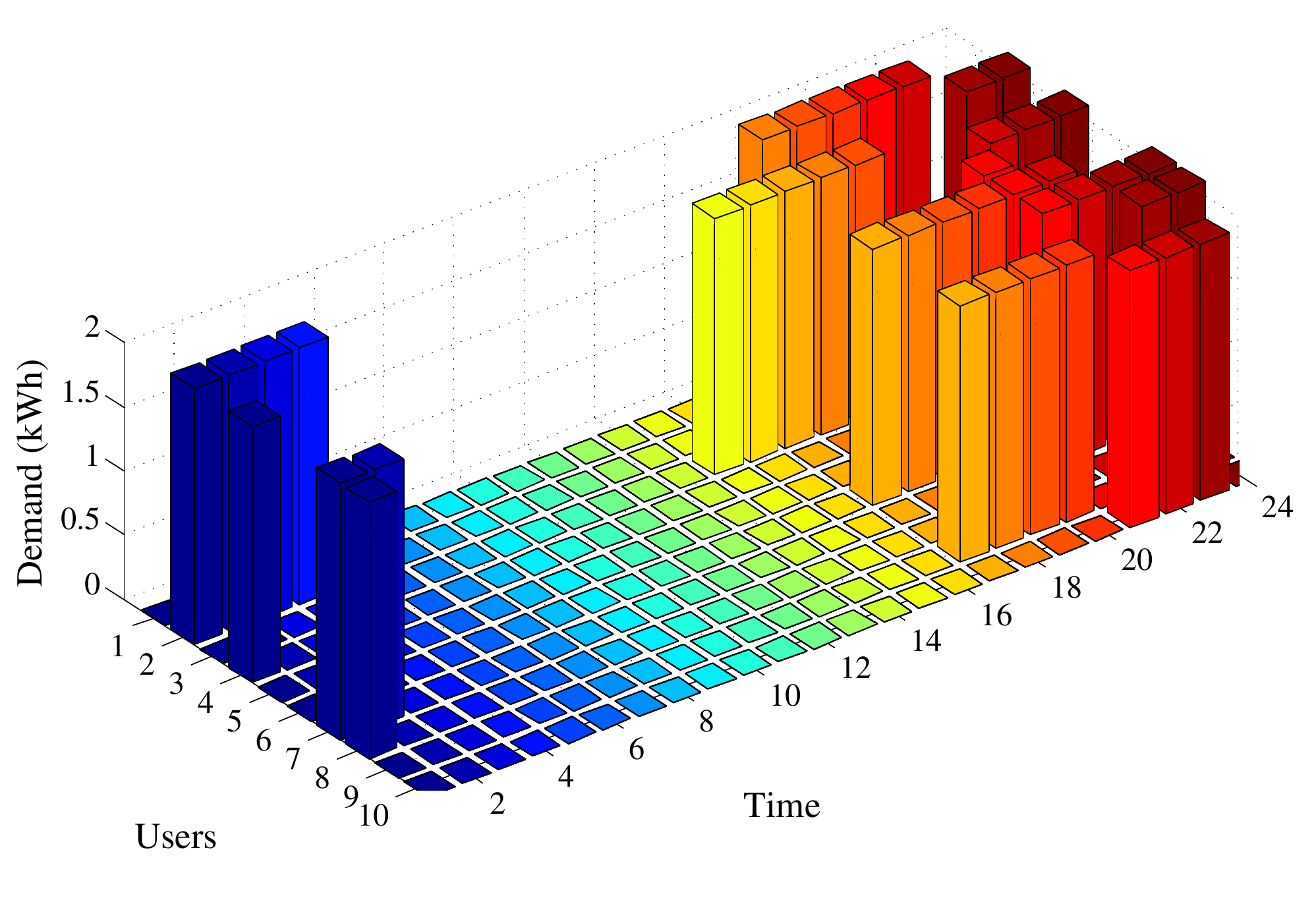}\vspace{-1em}
  	\caption{Some realizations of the stochastic process defined in (\ref{xt}) in modulo 24-hours.}
  	\label{f7}
  \end{figure}

Furthermore, a DR technique affects $x(t)$ and changes its statistics. This process can be modelled as if $x(t)$ is passed through a system as shown in Fig. \ref{f5}. Therefore, the information about the statistics of the input helps to design the system such that the resulting random process $y(t)$ fulfills the desired objectives of the DR techniques. The power consumption profile $y(t)$ results from both the DR algorithm and the particular statistics of the original power consumption profile.

\begin{figure}
	\centering
	\includegraphics[width=\columnwidth]{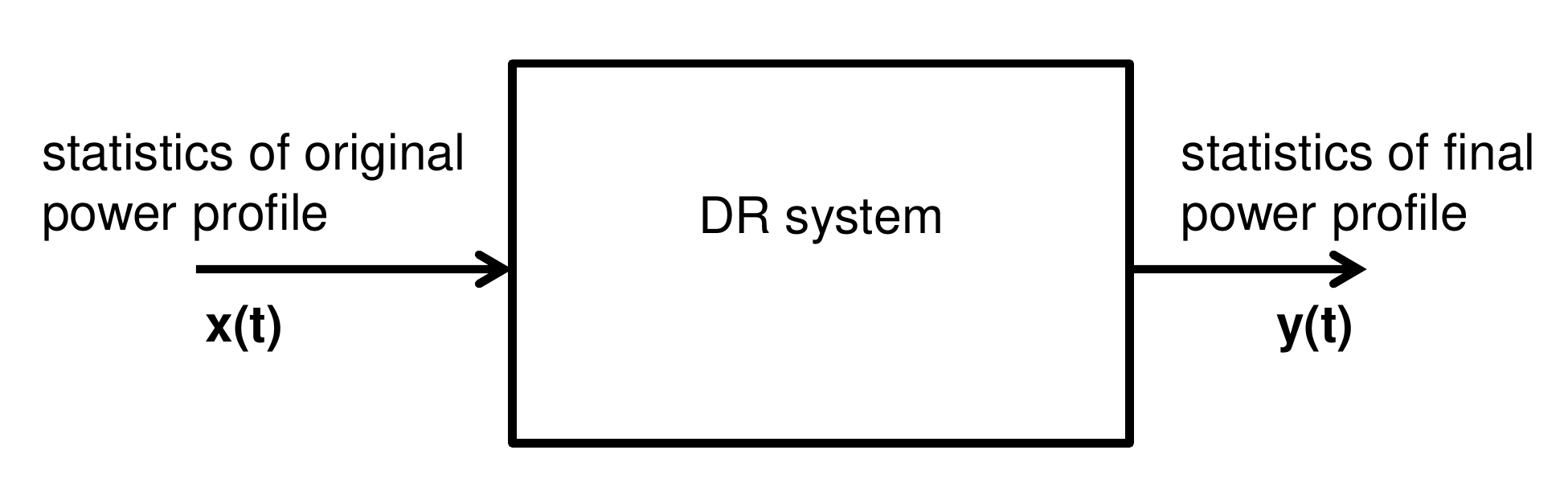}
	\caption{Demand response technique modelled as a system.}
	\label{f5}
\end{figure}

We can assume probability distribution functions (PDFs) for these two random variables, for instance, according to synthesized models obtained from experimental data, e.g., in \cite{lee_stochastic_2012} for PHEVs. Here, focusing on PHEVs, we assume $t_0$ and $T$ have independent PDFs that can be found from empirical data. For example, for $t_0$, as the arrival time, a Gaussian distribution is suggested in \cite{lee_stochastic_2012}:    

\begin{equation}
\begin{array}{lcr}
t_0 &\thicksim & N(\mu,{\sigma}^2)  
\end{array}
\label{dist}
\end{equation} 
 
\noindent{where $\mu$ and ${\sigma}^2$ denote the mean and variance of the Gaussian distribution, respectively. For PHEVs, there also exist different charging modes as described in Table \ref{tab1}. The charging mode may be considered related to the other random variable in (\ref{xt}) which is $a$. But, we note that there is a tight correlation between $a$ and $T$. This is obvious due to the fact that on fast charging modes the charging time $T$ is much shorter.}

\begin{table}[!t]
\renewcommand{\arraystretch}{1.3}
\caption{Different Types of Charging Outlets  (\lowercase{http://www.teslamotors.com/})}
\label{tab}
\centering
\begin{tabular}{|c||c| |c| |c|}
\hline
OUTLET & V/A &   kW & 
$
\begin{array}{c}
\text{MILES/1-HOUR}\\
 \text{OF CHARGING}
\end{array}
$
 \\ 
\hline
Standard & 110 / 12  & 1.4 kW & 3 \\
 \hline
Newer Standard &  110 / 15 & 1.8 kW &  4 \\
 \hline
Single Fast & 240 / 40	 & 10 kW & 29  \\
  \hline
 Twin Fast & 240 / 80	 & 20 kW & 58  \\
    \hline
\end{tabular}

\label{tab1}
\end{table}

\section{Statistical Analysis} \label{Sta}

In this section, using the aforementioned definition of $x(t)$, we calculate $\mathbb{E}[x(t)]$ which represents the expected value of power consumption for a certain appliance. This expectation can be expressed by the following proposition for PHEVs (refer to the appendix for the proof). 

\begin{proposition}
Given $f_{t_0}(\cdot)$ and $f_T(\cdot)$ as the PDFs of the independent random variables arrival time $t_0$ and charging time $T$ for a PHEV, the expected uncoordinated charging power demand can be expressed as:

\begin{gather}
 \mathbb{E}[x(t)]=a\times \big( F_{t_0}(t) \ast [\delta(t)-f_T(t)]  \big)
\label{EX}
\end{gather}
\noindent{in which, $\ast$ shows the convolution operation and $\delta(t)$ is the unit impulse function. Also, $F(\cdot)$ represents the cumulative distribution function (CDF).} 

\label{pro}
\end{proposition}
 
We can calculate (\ref{EX}) for any given distribution analytically or numerically. Hereafter, we adopt different distributions for the PHEV's charging time $T$ following some available empirical research data in the literature, as shown in Fig. \ref{besuni}, to study the corresponding results of (\ref{EX}). We investigate four cases for the distribution of $T$, namely, the uniform, exponential, Gaussian with positive support, and Rician distributions. These distributions have different degrees of freedom (DoF) and all of them support $T$ over $[0,+\infty)$:
\begin{itemize}
\item{\textbf{T: Uniform} } In this case, we consider $T$ to have uniform distribution over the interval $[c,d)$. Then, $\mathbb{E}[x(t)]$ can be analytically derived as stated in the following proposition (see the appendix for the proof). 
\end{itemize}

\begin{proposition} Assuming $t_0$ has a normal distribution with mean $\mu$ and variance $\sigma^2$ and $T$ has a uniform distribution over the interval $[c,d)$, $0\leq c < d$, the expected uncoordinated charging power demand becomes:  
\begin{gather}
\nonumber \mathbb{E}[x(t)]=a\times 
\bigg[ 1-\mathbf{Q}(\frac{t-\mu}{\sigma})+\frac{\sigma}{d-c}
( c'\mathbf{Q}(c')\\
-d'\mathbf{Q}(d')+f(d')-f(c')+d'-c' ) \bigg]
\end{gather}
\label{prou}

\noindent{where $c'=\frac{t-c-\mu}{\sigma}$, $d'=\frac{t-d-\mu}{\sigma}$. Also, $\mathbf{Q}(x)$ and $f(x)$ are defined as follows: }

\[ \mathbf{Q}(x)=\frac{1}{\sqrt{2\pi}}\int\limits_x^\infty \exp (-\frac{u^2}{2}) du,\]

\[ f(x)=\frac{\exp(-\frac{x^2}{2})}{\sqrt{2\pi}}.\]
\end{proposition}

\begin{figure}
      \centering
      \includegraphics[width=\columnwidth]{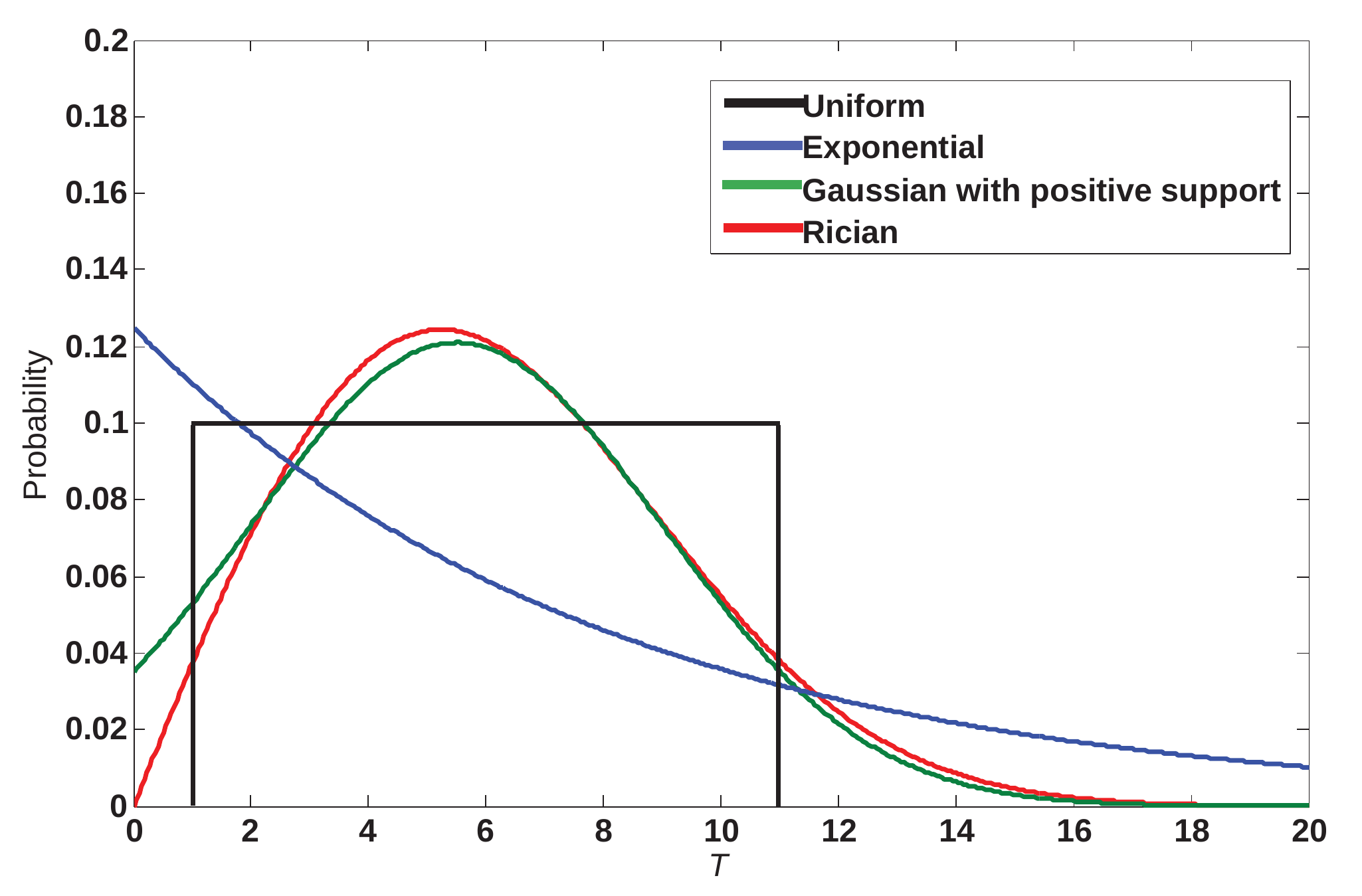}\vspace{-1em}
      \caption{Uniform, exponential, Gaussian with positive support and Rician distributions for $T$.}
      \label{besuni}
\end{figure}

\begin{itemize}

\item{\textbf{T: Exponential}} The driven distance and hence the charging time of an EV can be modelled by an exponential distribution \cite{liang_stochastic_2014}. For an exponentially distributed $T$ with mean $\lambda^{-1}$, we have the following PDF: 

\begin{gather}
f_T(T)=\lambda\exp(-\lambda T).
\label{pexp}
\end{gather}

\item{\textbf{T: Gaussian}} When $T$ has a Gaussian PDF with positive support as shown in Fig. \ref{besuni}, $T$ has the following distribution function: 

\begin{gather}
 f_T(T)=N(T;\mu,\sigma^2|0\leq T < \infty),\\
=\frac{1}{\mathbf{Q}(\frac{-\mu}{\sigma})\sqrt{2\pi \sigma^2}}\exp (-\frac{(T-\mu)^2}{2\sigma^2}),\, \,\,\, 0\leq T < \infty. 
\label{pgauss}
\end{gather}

%
%
%

\item{\textbf{T: Rician}} Finally, we consider a Rician PDF for $T$ having the following form:

\begin{equation}
f(T|\nu , \sigma)=\frac{T}{\sigma^2}\exp(-\frac{(T^2+\sigma^2)}{2\sigma^2})I_0(\frac{T\nu}{\sigma^2})
\label{rice}
\end{equation}

\noindent{where $\nu \geq 0$ and ${\sigma} \geq 0$  present the noncentrality parameter and scale parameter, respectively. $I_0(\cdot)$ is the modified Bessel function of the first kind with order zero.}

\end{itemize}

\par

\begin{figure}
	\centering
		\includegraphics[width=\columnwidth]{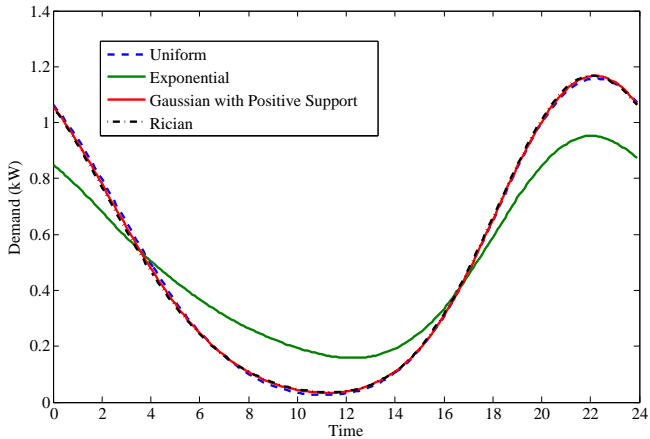}\vspace{-1em}
	\caption[]{A PHEV's expected daily power demand profile for different distributions of charging time $T$.}
	\label{com}
\end{figure}

\vspace{+1em}
\section{Simulation Results} \label{SR}
\vspace{+1em}
In this section, we consider Gaussian distribution for the random variable $t_0$ as the arrival time with $\mu=19$ and $\sigma^2 = 10$ inspired from \cite{lee_stochastic_2012}. Furthermore, we consider four cases for the distribution of the random variable $T$ as described in section \ref{Sta}. First, we consider $T$ to have a uniform distribution over the interval $[1,11]$. Thus, it will have $\mu=6$ and $\sigma^2=8.33$. Second, we assume $T$ to be exponentially distributed with mean $\mu=6$. Third, we assume $T$ to be Gaussian distributed with positive support as presented in (\ref{pgauss}). In this case, we use the well-known \textit{accept-reject} approach to generate the random values. 
Finally, we consider a Rician distribution for $T$. In all cases (except for the exponential distribution), we set the parameters of the distributions such that they all have the same mean and variance. However, for the exponential distribution case, we can only set either its mean or variance to be the same as that of the others since this distribution has just one DoF. 
Based on an average 0.25 kWh energy consumption for each mile of driving, we set all the parameters in (\ref{xt}). 
In addition, we assume a system comprising of $N=100,000$ PHEV users in our simulations in order to obtain smooth curves representing the probabilistic expectation. \par

The results for the expected daily power demand of a typical PHEV under the aforementioned settings are illustrated in Fig. \ref{com}. As can be observed, the expected daily power demand resulting from the charging time distributions which possess the same mean and variance tends to the same power profile. However, for the exponential distribution, since it has only one DoF, we see that its expected power demand differs significantly from that of the others. \par 

Based on our proposed model and the obtained results, we observe that the expected uncoordinated charging power demand for a typical PHEV is much larger during 6 p.m. to 1 a.m. compared to that during 8 a.m. to 2 p.m. in a one-day frame. \par

Next, in Fig. \ref{uniex}, we compare our obtained analytical result for the expected power demand according to a uniform distribution of $T$ in proposition (\ref{prou}) with the simulation results in a one-day frame. As can be seen in this figure, the simulation results follow proposition (\ref{prou}) closely, affirming the acquired formulation.  \par

  \begin{figure}
  	\centering
  	\includegraphics[width=\columnwidth]{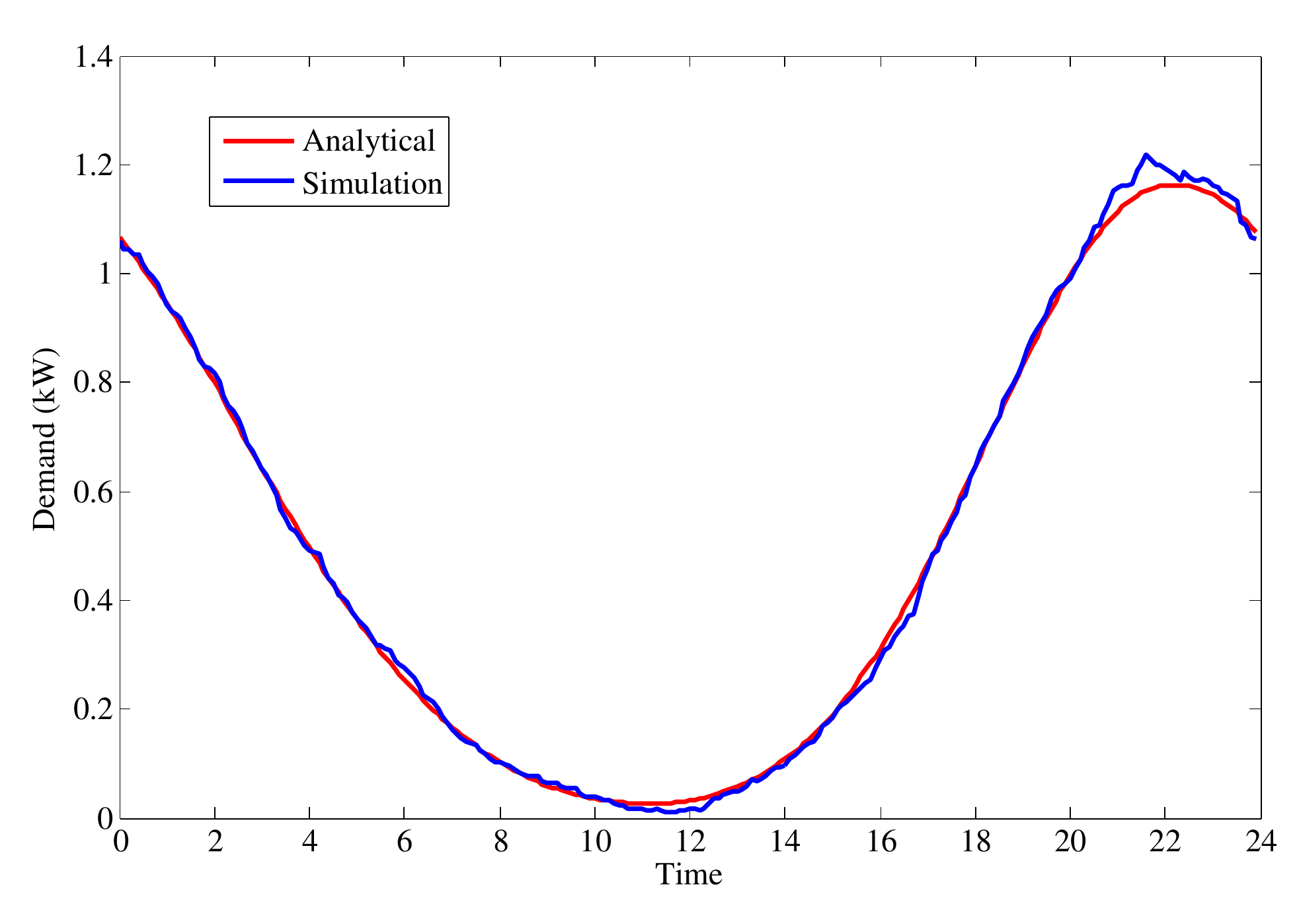}\vspace{-1em}
  	\caption{A PHEV's expected daily power demand profile for a uniform distribution of charging time $T$.}
  	\label{uniex}
  \end{figure}
\vspace{2em}
\section{Conclusion and Future Work} \label{Con}

In this paper, we discussed the inherent randomness in the demand for flexible appliances in general and for PHEVs in particular. We considered random distributions for the arrival time and the charging time of PHEVs inspired by available empirical data in the literature. Accordingly, we presented the uncoordinated charging power demand impact of a PHEV as a stochastic process based on these random variables. Next, we derived the expected daily power consumption profile according to this random process. Our simulation results show that the EV's expected daily power demand profiles obtained under the uniform, Gaussian with positive support and Rician distributions for charging time are identical when the first and second order statistics of these distributions are the same. Our obtained results introduce a simple description for the expected power demand of a typical PHEV and hence give us insights into the effect of adding each PHEV into the power system.  

\par 

The study presented in this paper can be extended and developed in various ways. For example, the convergence of the daily power demand for the aforementioned distributions needs to be proven. In addition, the results from this modelling can be incorporated into designing DR algorithms and evaluating the available DR techniques more accurately.

\appendix

%

\begin{itemize}

\item[A.]
\begin{proof}[Proof of proposition \ref{pro}]
Since $x(t)=0$ for $t_0\leq t-T$ and $t\leq t_0$. Then, $\mathbb{E}[x(t)]$ becomes:

\begin{gather}
\nonumber \mathbb{E}[x(t)]=a\times P(t_0 \leq t \leq t_0+T)\\ 
=a\times P(t-T \leq t_0 \leq t).  
\end{gather}
Further, we can use the \textit{total probability theorem} \cite{kobayashi2012probability} to get

\begin{gather}
\nonumber \mathbb{E}[x(t)]= a\times\int\limits_0^\infty P(t-T \leq t_0 \leq t | T=T') f_T(T')dT'\\
 = a\times \int\limits_0^\infty (F_{t_0}(t) - F_{t_0}(t-T')) f_T(T')dT' \\ 
 =a\times \left[F_{t_0}(t)-\int\limits_0^\infty F_{t_0}(t-T')f_T(T')dT'\right]
\label{int}
\end{gather}
 
for which we have taken into account the facts that $\int\limits_0^\infty f_T(T')dT'=1$, and $t_0$ and $T$ are independent. Furthermore, we can express (\ref{int}) in a more concise form by using the definition of the convolution integral and the identity $f(t) \ast \delta(t)=f(t)$ as follows:   

\begin{gather}
 \mathbb{E}[x(t)]=a\times \left( F_{t_0}(t) \ast [\delta(t)-f_T(t)]  \right).
\end{gather}

\end{proof}

\item[B.]
\begin{proof}[Proof of proposition \ref{prou}]
Since $T$ is uniformly distributed over the interval $[c,d)$, $0\leq c < d$, we can write (\ref{int}) as follows:

\begin{gather}
\mathbb{E}[x(t)]=a\hspace{-0.8mm}\times \hspace{-1.2mm} \left[\hspace{-0.5mm}F_{t_0}(t)-\frac{1}{d-c}\int\limits_c^d \hspace{-1.5mm} F_{t_0}(t-T')dT'\hspace{-0.2mm}  \right].
\label{intuni}
\end{gather}
Then, by changing the integration variable from $T'$ to $\alpha=t-T'$, it can be rewritten as follows: 

\begin{gather}
\mathbb{E}[x(t)]=a\times \left[ F_{t_0}(t)+\frac{1}{d-c}\int\limits_{t-c}^{t-d} F_{t_0}(\alpha)d\alpha \right].
\label{intuni2}
\end{gather}
Further, we need to replace $\alpha$ with $\beta=\frac{\alpha-\mu}{\sigma}$ to have 

\begin{gather}
\mathbb{E}[x(t)]=a\times \left[ F_{t_0}(t)+\frac{\sigma}{d-c}\int\limits_{\frac{t-c-\mu}{\sigma}}^{\frac{t-d-\mu}{\sigma}} F_{t_0}(\beta)d\beta \right]
\label{intuni3}
\end{gather}

in order to be able to use the following formula for a standard normal random variable with CDF $F(\cdot)$ and PDF $f(\cdot)$ to calculate the last term in (\ref{intuni3}): 

\begin{gather}
\int F(x)dx=xF(x)+f(x)+c.
\label{intnorm}
\end{gather}

Also, we now set $c'=\frac{t-c-\mu}{\sigma}$ and $d'=\frac{t-d-\mu}{\sigma}$ for simplicity to express (\ref{intuni3}) in the following form:

\begin{align}
\mathbb{E}[x(t)]=  a\times \bigg[  1-\mathbf{Q}(\frac{t-\mu}{\sigma})+\frac{\sigma}{d-c}(c'\mathbf{Q}(c')  \nonumber \\
  -d'\mathbf{Q}(d')+f(d')-f(c')+d'-c') \bigg]
\label{intuni4}
\end{align}
in which we used the equation $F(x)=1-\mathbf{Q}(x)$. 
\end{proof}
\end{itemize}

\bibliographystyle{IEEEtran} 
\bibliography{IEEEabrv,myBIB}

\end{document}